\documentclass[11pt]{article}
\usepackage[utf8]{inputenc}
\usepackage[T1]{fontenc}
\usepackage{lmodern}
\usepackage[letterpaper,margin=1in]{geometry}
\usepackage{amsmath,amssymb,amsthm,mathtools}
\usepackage{adjustbox}
\usepackage{tikz}
\usetikzlibrary{arrows.meta,positioning}
\usepackage{microtype}
\usepackage{needspace}
\usepackage{booktabs,array}
\usepackage{enumitem}
\usepackage{xcolor}
\usepackage{hyperref}
\hypersetup{colorlinks=true,linkcolor=blue!55!black,citecolor=blue!55!black,
  urlcolor=blue!55!black,pdftitle={Finite-Sample Dependence Measures under Data Processing},
  pdfauthor={Yuqing Kong}}
\setlist[itemize]{topsep=4pt,itemsep=3pt}
\setlist[enumerate]{topsep=4pt,itemsep=3pt}
\newtheorem{theorem}{Theorem}[section]
\newtheorem{lemma}[theorem]{Lemma}
\newtheorem{proposition}[theorem]{Proposition}
\newtheorem{corollary}[theorem]{Corollary}
\newtheorem*{boundarysummary}{Theorem A (zero-row vanishing \& tall-case collapse)}
\newtheorem*{algebrasummary}{Theorem B (determinantal rigidity)}
\theoremstyle{definition}
\newtheorem{definition}[theorem]{Definition}
\theoremstyle{remark}
\newtheorem{remark}[theorem]{Remark}

\DeclareMathOperator{\rank}{rank}

\newcommand{\bbR}{\mathbb R}

\newcommand{\ones}{\mathbf 1}
\newcommand{\QMI}{F_{\mathrm{QMI}}}
\newcommand{\MQMI}{\mathcal M_{\mathrm{QMI}}}
\newcommand{\ind}{\mathbf 1}
\title{Limits of Finite-Sample Dependence Measures\\under Data Processing Inequality}
\author{Yuqing Kong\\Peking University}
\date{}
\begin{document}
\pagenumbering{gobble}
\maketitle
\begin{abstract}
A measure of dependence assigns a real number to the joint distribution of two random variables, with value zero when the variables are independent. A natural requirement is the data processing inequality (DPI): randomly post-processing one variable should not increase its dependence on the other. We distinguish
\emph{fixed-alphabet DPI}, which preserves the alphabet size of
the processed variable, from \emph{dimension-changing DPI},
which allows this size to change. Another desirable property is unbiased estimability from a fixed finite number of i.i.d. samples, since in practice we may observe only samples from the joint distribution. Mutual information vanishes at independence and satisfies the data processing inequality, but admits no such unbiased estimator.

We ask which dependence measures can satisfy DPI while still admitting a finite-sample unbiased estimator. When the two random variables each take finitely many possible values, such an unbiased estimator exists if and only if the dependence measure is a polynomial in the probabilities specifying their joint distribution. Therefore, the question becomes which polynomial dependence measures satisfy DPI. 

For \emph{fixed-alphabet DPI}, suppose the processed variable
has \(n\) states and the other variable has \(m\) states, where
\(n,m\ge2\). If \(n>m\), we prove an impossibility result that
does not require polynomiality: every dependence measure
satisfying one-sided fixed-alphabet DPI and admitting a
continuously differentiable extension to a neighborhood of the
closed probability simplex must vanish identically.
If \(n\le m\), every nonzero polynomial dependence measure
satisfying the same DPI condition has degree at least \(2n\),
and this bound is tight.

For \emph{dimension-changing DPI}, the impossibility holds
across all alphabet sizes: every family of dependence measures
satisfying one-sided dimension-changing DPI and admitting such
continuously differentiable extensions is identically zero.
In particular, no nonzero family of polynomial dependence
measures satisfies dimension-changing DPI.

The fixed-alphabet results yield sharp task bounds for multi-task peer prediction. In this setting, two agents receive correlated private signals across independent tasks, such as movie rating, and are rewarded based on their reports. Expected payments from finitely many tasks are polynomial in the joint signal probabilities. Requiring zero expected payment for independent reports gives vanishing at independence, while requiring that stochastic post-processing of an agent's signal cannot increase her expected payment imposes one-sided fixed-alphabet DPI. When both agents have $n$ possible signals, $2n$ tasks are necessary and sufficient for nonzero expected-payment functionals satisfying these requirements. When their signal-space sizes differ, the expected-payment functional of the agent with the larger space must vanish identically. Thus, no finite-task mechanism satisfying these requirements can provide nontrivial expected payments to both agents.
\end{abstract}
\clearpage
\pagenumbering{arabic}
\section{Introduction}\label{sec:intro}

Can a measure of statistical dependence satisfy the data processing
inequality (DPI) while admitting an unbiased estimator from a fixed
number of samples? We show that these two natural requirements impose
sharp restrictions on the measure.

Let \(U\) be the joint-distribution matrix of finite-valued random
variables \(X\) and \(Y\), with alphabet sizes \(n,m\ge2\),
respectively\footnote{For binary variables, for example,
\(
U=
\begin{pmatrix}
\Pr[X=0,Y=0] & \Pr[X=0,Y=1]\\
\Pr[X=1,Y=0] & \Pr[X=1,Y=1]
\end{pmatrix}.
\)}. A dependence measure assigns a value \(F(U)\) to
this joint distribution. At a minimum, it should vanish when
\(X\) and \(Y\) are independent, or equivalently, when \(U\)
has rank one. A further natural requirement is the
\emph{data processing inequality}: randomly post-processing one
variable should not increase its dependence on the other.
Suppose \(X'\) is generated from \(X\) through a channel \(T\),
where \(T_{ij}=\Pr(X'=i\mid X=j)\), independently of \(Y\)
conditional on \(X\). The joint-distribution matrix of
\((X',Y)\) is then \(TU\), and the left-sided DPI requires
\[
    F(TU)\le F(U).
\]
When \(T\) is an \(n\times n\)
column-stochastic matrix, the channel preserves the alphabet
size of \(X\). We call this the \emph{fixed-alphabet} DPI. When \(T\) is rectangular, we call this \emph{dimension-changing} DPI. 

Shannon mutual information \cite{Shannon1948} satisfies DPI and vanishes exactly
at independence. However, in practice, we do not know \(U\). We observe only \(d\) independent
pairs
\(
(X_1,Y_1),\ldots,(X_d,Y_d)\stackrel{\mathrm{i.i.d.}}{\sim}U.
\) This raises a
second requirement: can a statistic of a fixed number of samples
have expectation exactly equal to the dependence measure for
every joint distribution? Shannon mutual information admits no
such unbiased estimator~\cite{Paninski2003}.

When $X$ and $Y$ are finite-valued, this estimation requirement has a simple
algebraic characterization: a functional of \(U\) admits an unbiased
estimator from \(d\) samples if and only if it is represented
on the probability simplex by a polynomial of degree at most
\(d\). Indeed, the probability of any sequence of \(d\)
observations is a product of \(d\) entries of \(U\), so the
expectation of any statistic is such a polynomial. Conversely,
each monomial of degree at most \(d\) has an unbiased estimator
obtained by assigning its factors to distinct samples.

Our question therefore becomes:
\emph{which polynomials vanish at independence and satisfy DPI?}

A natural polynomial candidate is the
discrete analogue of \emph{quadratic mutual information}~\cite{BerrettSamworth2021USP,Torkkola2003}:
\[
    F_{\mathrm{QMI}}(U)
    :=
    \bigl\|U-r(U)c(U)^\top\bigr\|_F^2,
\]
where \(r(U)\) and \(c(U)\) are the row and column marginals.
This measure vanishes exactly at independence. Since both
marginals are linear in \(U\), it is a polynomial of degree
at most four and admits an unbiased estimator from four samples.
However, it does not satisfy DPI in general (see a counter-example in Section~\ref{subsec:four-task}).

The \emph{squared determinant-based mutual information}
(squared DMI) shows that DPI and finite-sample unbiased
estimability can nevertheless hold
simultaneously~\cite{DBLP:conf/soda/Kong20}.
When both variables have \(n\) states, the functional
\[
    F_{\mathrm{DMI}}(U)=(\det U)^2
\]
is a degree-\(2n\) polynomial that vanishes at independence.
For every column-stochastic \(T\),
\[
    F_{\mathrm{DMI}}(TU)
    =(\det T)^2F_{\mathrm{DMI}}(U)
    \le F_{\mathrm{DMI}}(U),
\]
because \(|\det T|\le1\). Hence squared DMI satisfies
fixed-alphabet DPI and can be estimated unbiasedly from
\(2n\) samples.

However, this construction has three limitations. First, it vanishes
on every rank-deficient matrix. For \(n\ge3\), these include
dependent distributions\footnote{For instance,
\[
U=\frac1{18}
\begin{pmatrix}
3&1&2\\
1&3&2\\
2&2&2
\end{pmatrix}
\]
has uniform marginals and is not independent, yet its third row
is the average of its first two rows. Thus
\(\operatorname{rank}U=2\) and
\(F_{\mathrm{DMI}}(U)=0\).
}.
Second, the construction uses \(2n\) samples, a number that grows
with the alphabet size. Third, the determinant is defined only
for square matrices, so the construction does not directly
handle variables with different alphabet sizes.

Are these limitations specific to the determinant construction,
or are they forced by polynomiality and DPI?
Must every polynomial dependence measure satisfying these
conditions vanish on rank-deficient distributions?
Is the \(2n\)-sample requirement unavoidable for a nonzero
functional?
What changes when the two variables have different alphabet
sizes, so that \(U\) is rectangular?
We answer these questions through a common structural obstruction.

The same question arises in multi-task peer prediction
\cite{dasgupta2013crowdsourced}, a mechanism-design problem
motivated by platforms that want to elicit honest subjective
judgments, such as movie ratings, without access to ground truth.
A principal assigns a set of tasks to two agents, Alice and Bob.
For each task, the agents receive correlated private signals,
and their signal pairs are modeled as i.i.d.\ samples from
a joint distribution of two random variables \((X,Y)\).
The agents need not have the same signal space: for example,
Alice may provide a five-star rating while Bob provides a
binary rating. A mechanism observes their reports on finitely
many tasks, but not the underlying signal distribution or the
ground truth, and rewards them based on these reports.

We consider reporting strategies that apply the same stochastic
channel independently on each task. The reported pairs then
remain i.i.d., so the mechanism's expected payment is a polynomial
in the joint report probabilities, with degree bounded by the
number of tasks. Requiring zero expected payment for independent
reports gives vanishing at independence. Requiring that an agent
cannot increase her expected payment by stochastically garbling
her signal imposes one-sided fixed-alphabet DPI. Our structural results
therefore translate directly into limits of multi-task peer
prediction.

\subsection{Main results}

Write \(\Delta_{n,m}\) for the simplex of nonnegative \(n\times m\)
matrices with total mass one, where \(n,m\ge2\). Rows index the variable
being processed. Our default requirement is \emph{left-sided,
fixed-alphabet DPI}:
\[
   F(TU)\le F(U)
   \qquad
   (U\in\Delta_{n,m},\ T\in\mathbb R^{n\times n}
   \text{ column-stochastic}).
\]
We also require \emph{independence-zero}: \(F(U)=0\) whenever
\(\rank U=1\). These two conditions already imply \(F\ge0\).

The first result does not require polynomiality. A functional is
\emph{\(C^1\)-extendable} if it is the restriction of a \(C^1\) function
on an open neighborhood of the \emph{closed} simplex. Every polynomial
has this property.

\begin{boundarysummary}
An independence-zero, \(C^1\)-extendable functional satisfying left-sided
fixed-alphabet DPI vanishes on every joint distribution matrix \(U\) with a zero row.
If \(n>m\), it vanishes on all of \(\Delta_{n,m}\).
\end{boundarysummary}

When \(n>m\), every distribution is a stochastic image of one supported
on only \(m\) rows, so zero-row vanishing propagates to the whole
simplex. Though Shannon mutual information satisfies interior smoothness, it escapes the obstruction through its non-smooth boundary behavior. 

For the remaining alphabet sizes, polynomiality yields a sharp
algebraic statement. Put \(\tau(U)=\sum_{i,\alpha}U_{i\alpha}\), and
let \(I_n\) be the ideal generated by the \(n\times n\) minors of \(U\).
Its square consists of polynomial combinations of products of two
maximal minors.

\Needspace{9\baselineskip}
\begin{algebrasummary}
Suppose \(2\le n\le m\). Every independence-zero polynomial \(F\)
satisfying left-sided fixed-alphabet DPI obeys
\[
   F\in I_n^2+(\tau-1).
\]
If its restriction to \(\Delta_{n,m}\) is nonzero, then
\(\deg F\ge2n\). The bound is attained by \(\det(UU^\top)\).
For \(n=m\), the conclusion is
\(F(U)=(\det U)^2G(U)\) on the simplex for some polynomial \(G\).
\end{algebrasummary}

The polynomial consequences are summarized below.
\begin{center}
\renewcommand{\arraystretch}{1.2}
\begin{tabular}{@{}lll@{}}
\toprule
Alphabet sizes & Necessary form on the simplex & Minimum sample count\\
\midrule
\(n>m\) & \(F=0\) & No nonzero functional\\
\(n=m\) & \(F=(\det U)^2G\) & \(2n\)\\
\(n<m\) & \(F\in I_n^2\pmod{\tau-1}\) & \(2n\)\\
\bottomrule
\end{tabular}
\end{center}

\paragraph{Consequences for information measures.}
The zero-row theorem has two further implications
(Section~\ref{sec:consequences}). First, every
\(C^1\)-extendable, independence-zero family satisfying
dimension-changing DPI is identically zero. Indeed, adding an
unused state is reversible by a channel and therefore preserves
the functional's value, while the resulting zero row forces that
value to be zero.

Second, for \(n,m\ge2\), there exists a \(C^1\)-extendable
functional satisfying left-sided fixed-alphabet DPI and vanishing
\emph{exactly} at independence (i.e., it is \emph{independence-faithful}) if and only if \(n=2\).
Under two-sided fixed-alphabet DPI, a finitely unbiasedly estimable functional
vanishing exactly at independence exists if and only if
\(n=m=2\). 

\paragraph{Consequences for peer prediction.}
An \(\ell\)-task payment rule has an expected-payment polynomial of
degree at most \(\ell\). We require the expected payment to satisfy one-sided fixed-alphabet DPI such that truthful reporting yields at least as much expected payment as applying the same independent stochastic garbling for each task. We additionally require the expected payment to vanish on independent reports. Our
results imply that \(2n\) tasks are necessary and sufficient for
nontrivial expected payments when both alphabets have size \(n\).
For unequal sizes\footnote{This happens when we ask Alice to provide five-star ratings and Bob to provide binary ratings.}, the larger-alphabet agent's expected payment is
identically zero for every finite-task mechanism satisfying these
conditions (Section~\ref{sec:peer-prediction}). 

Section~\ref{sec:restricted} gives a complementary
four-task construction for arbitrary alphabets under deviations that
mix relabelings and signal-independent reports. It uses the quadratic mutual information as its payment target.

\subsection{Related work}\label{sec:related}
\paragraph{Multi-task peer prediction and task complexity.}
Dasgupta and Ghosh~\cite{dasgupta2013crowdsourced} introduce the multi-task peer-prediction setting and propose a mechanism for eliciting truthful binary signals (e.g., yes-or-no reports). Correlated Agreement~\cite{2016arXiv160303151S} extends this approach to non-binary signals (e.g., ratings from one to five).

A central question is how many tasks are needed for different
incentive guarantees. Under structural assumptions on the joint
distribution, two tasks suffice for weaker notions of truthfulness
\cite{dasgupta2013crowdsourced,2016arXiv160303151S}.
The Mutual Information Paradigm
\cite{Kong:2019:ITF:3309879.3296670} connects truthful reporting
to dependence measures satisfying DPI: under dominant truthfulness,
an agent cannot increase her expected payment by stochastically
garbling her signal. Shannon mutual information provides such a
reward in principle, but its non-polynomial form prevents an exact
unbiased implementation with any fixed finite number of tasks.
Related approaches obtain guarantees by learning the required
statistics from an infinite number of tasks~\cite{2016arXiv160303151S,liuchen}.

The Determinant-based Mutual Information (DMI) mechanism
\cite{DBLP:conf/soda/Kong20} achieves dominant truthfulness with
\(2n\) tasks when both agents' reports have \(n\)-symbol alphabets. A signal's alphabet is its set of possible
values: a yes-or-no answer has a binary alphabet, whereas a rating
from one to five has an alphabet of size five. The DMI mechanism's expected payment is proportional to \((\det U)^2\), and hence is a polynomial satisfying DPI and admitting a \(2n\)-sample unbiased estimator. Volume Mutual Information (VMI)
\cite{DBLP:conf/innovations/Kong22} expands the choice of
information-monotone rewards and finite-task mechanisms, in part
to improve effort incentives that DMI need not optimize. Known
polynomial constructions in this line use square joint-distribution
matrices and retain a determinant-square factor. Our results show
that these restrictions reflect a general obstruction: \(2n\) is
the minimum task count for a nontrivial polynomial DPI reward with
two agents having equal \(n\)-symbol alphabets. When the two agents' reports have different alphabet sizes—for example, one provides a five-star rating while the other provides a binary rating—the reward of the agent with the larger alphabet vanishes. 

The incentive guarantees studied here concern memoryless deviations, in which an agent applies the same stochastic reporting rule independently to each task, using only her signal for that task. Recent work~\cite{Fro26} studies joint-task deviations, which allow an agent's reports to depend on her entire signal vector, and shows that truthful reporting remains a best response under the DMI mechanism among all joint-task strategies when the other agents use consistent, memoryless strategies. Our lower bounds also apply under the stronger requirement of truthfulness against joint-task deviations, since this broader class includes all memoryless deviations.

Recent work~\cite{vonallmen2026sample} shows that, under two-sided fixed-alphabet DPI and independence-zero, no nontrivial functional on binary joint distributions admits a fixed-sample unbiased estimator using three or fewer samples. Every such nontrivial functional admitting an unbiased estimator from four or five samples is a positive scalar multiple of squared DMI, while this uniqueness up to scaling fails at six samples. It
also derives sample-order-invariant, convex-minimal estimators for DMI and
analyzes their variance. The paper leaves the corresponding structural questions for larger and unequal
alphabets open, and cites an earlier version of the present work%
\footnote{An earlier version of this paper circulated as
\emph{The Algebraic Limits of Polynomial Information Measures}.}
as resolving them. Other peer-prediction impossibilities impose geometric conditions on
specified distribution families~\cite{geometricpp,zheng2021limits};
our bounds require truthfulness uniformly over all joint
distributions.

\paragraph{Information measures and estimation.}
Shannon mutual information and many standard
\(f\)-mutual informations satisfy DPI
\cite{Shannon1948,1571417124017192704}.
However, Shannon mutual information and other familiar examples
are not polynomial in the joint probabilities and therefore
admit no exactly unbiased estimator from a fixed number of samples.
Work on entropy and related functionals studies estimators
based on polynomial approximation
\cite{Paninski2003,JiaoVenkatHanWeissman,WuYang2016}
or on reconstructing features of the underlying distribution
\cite{ValiantValiant2017}. Instead of constructing polynomial approximations of the mutual information, we ask which polynomial dependence measures themselves satisfy DPI.

\paragraph{Garbling and determinantal structure.}
With the column marginal fixed, left-sided DPI expresses
monotonicity under Blackwell garbling~\cite{Blackwell1953}.
Recent work studies differential characterizations of monotone
information costs~\cite{cheng2026monotonicityinformationcosts}. We focus on studying the smoothness at the boundary of the probability
simplex. We show that if a dependence measure satisfying DPI and
independence-zero admits a \(C^1\) extension across the boundary
of the probability simplex, then it vanishes on every zero-row face.

Classical invariant theory identifies maximal minors as
generators of the ring of \(\mathrm{SL}_n\)-invariant polynomials
\cite{Weyl1939ClassicalGroups,Richman1989VectorInvariants}.
Our result does not assume \(\mathrm{SL}_n\)-invariance.
Here, the data processing inequality---a form of stochastic
monotonicity---and independence-zero force every admissible
polynomial into the square of the maximal-minor ideal.

Determinant-based constructions have also appeared in robust learning and data evaluation. A log-determinant loss has been used to obtain
robustness to label noise \cite{10.5555/3454287.3454846}, while more recent work uses Gram determinants to score dataset
reliability and establishes preservation of several reliability orderings \cite{ChenEtAl2026DataReliability}. Both constructions exploit the multiplicative structure of determinants to obtain
robustness or ordering guarantees.

\paragraph{Organization.}
Sections~\ref{sec:setup}--\ref{sec:non-tall} develop the model and
prove the zero-row vanishing, tall-case collapse and algebraic theorems.
Sections~\ref{sec:consequences} and \ref{sec:peer-prediction} give
their information-measure and mechanism-design consequences.
Section~\ref{sec:restricted} treats restricted deviations.
The appendices give the determinant fact used in Section~\ref{sec:remove-pivot},
a broad family of polynomial information measures, one-sided and two-sided normalization
and independence tests, and coNP-hardness of recognizing DPI for
circuit-represented polynomials.

\section{Set up}\label{sec:setup}

Throughout, \(n,m\ge2\). For a matrix \(U\in\mathbb R^{n\times m}\),
write
\[
 \tau(U):=\sum_{i,\alpha}U_{i\alpha},\qquad
 \Delta_{n,m}:=\{U\ge0:\tau(U)=1\}.
\]
The row and column marginals are \(r(U)=U\ones_m\) and
\(c(U)=U^\top\ones_n\). Independence is equivalent to
\(U=r(U)c(U)^\top\), or to \(\rank U=1\). Latin indices refer to
rows, Greek indices to columns. 

A matrix \(T\) is \emph{column-stochastic} if its entries are
nonnegative and each column sums to one. Left multiplication by an
\(n\times n\) column-stochastic matrix preserves \(\Delta_{n,m}\).
For a functional \(F:\Delta_{n,m}\to\mathbb R\), we use the conditions
\begin{align}
 F(TU)&\le F(U)
 &&\text{for every column-stochastic }T\in\mathbb R^{n\times n},
 \tag{LF-DPI}\label{eq:dpi}\\
 F(U)&=0
 &&\text{whenever }\rank U=1.
 \tag{IZ}\label{eq:iz}
\end{align}
Both conditions must hold for every \(U\in\Delta_{n,m}\).
LF-DPI stands for \emph{left-sided fixed-alphabet data processing inequality}.
\emph{Two-sided} fixed-alphabet DPI additionally requires
\(
F(US^\top)\le F(U)
\)
for every \(m\times m\) column-stochastic matrix \(S\). IZ stands for \emph{independence zero}: on the simplex \(\Delta_{n,m}\),
rank-one matrices correspond precisely to joint distributions
of independent random variables.

\begin{lemma}[Nonnegativity and permutation invariance]
\label{lem:refPone-derived}
Conditions \eqref{eq:dpi} and \eqref{eq:iz} imply \(F\ge0\).
Condition \eqref{eq:dpi} alone implies \(F(QU)=F(U)\) for every row
permutation matrix \(Q\).
\end{lemma}
\begin{proof}
A rank-one channel \(T=q\ones_n^\top\), where \(q\) is a probability
vector, sends \(U\) to the independent distribution \(qc(U)^\top\).
Thus \(0=F(TU)\le F(U)\). Applying DPI with both \(Q\) and
\(Q^{-1}\) gives permutation invariance.
\end{proof}

\begin{definition}[Boundary regularity]\label{def:C1-functional}
A functional \(F:\Delta_{n,m}\to\mathbb R\) is
\emph{\(C^1\)-extendable} if there is an open set
\(\Omega\subseteq\mathbb R^{n\times m}\) containing
\(\Delta_{n,m}\) and a continuously differentiable function
\(\widetilde F:\Omega\to\mathbb R\) agreeing with \(F\) on the simplex.
\end{definition}

Every polynomial is \(C^1\)-extendable. The condition also permits
non-polynomial functions, but is stronger than smoothness on the
strictly positive simplex. Remark~\ref{rem:shannon-boundary-derivative}
shows explicitly why Shannon mutual information does not satisfy it.

\begin{definition}[Fixed-sample unbiased estimability]
\label{def:finite-sample-estimability}
For a finite set \(\mathcal Z\), a functional
\(F:\Delta(\mathcal Z)\to\mathbb R\) is \emph{\(d\)-sample unbiasedly
estimable} if there is a statistic \(\widehat F:\mathcal Z^d\to\mathbb R\)
such that \(\mathbb E_P[\widehat F(Z_1,\ldots,Z_d)]=F(P)\) for every
\(P\), with \(Z_1,\ldots,Z_d\) drawn i.i.d.\ from \(P\).
It is \emph{finitely unbiasedly estimable} if this holds for some finite
\(d\).
\end{definition}

\begin{proposition}[Polynomial characterization~\cite{Hoeffding01091984}]
\label{prop:finite-sample-polynomial}
A functional on \(\Delta(\mathcal Z)\) is \(d\)-sample unbiasedly
estimable if and only if it is represented on that simplex by a
polynomial of degree at most \(d\) in the probabilities \(P(z)\).
\end{proposition}
\begin{proof}
For any statistic on \(d\) samples,
\[
 \mathbb E_P[\widehat F]
 =\sum_{z_1,\ldots,z_d\in\mathcal Z}
   \widehat F(z_1,\ldots,z_d)\prod_{t=1}^dP(z_t),
\]
which is a polynomial of degree at most \(d\). Conversely, each
monomial \(\prod_zP(z)^{a_z}\) of total degree \(k\le d\) is the
expectation of \(\prod_{t=1}^k\mathbf1\{Z_t=w_t\}\), where the
sequence \(w_1,\ldots,w_k\) contains \(a_z\) occurrences of \(z\). Thus, every polynomial admits an unbiased estimator, since it is a linear combination of monomials. The resulting estimator requires as many samples as the polynomial’s degree, the highest degree among its monomials.
\end{proof}

\paragraph{Polynomial functions on the simplex.}
Let \(R_{n,m}=\mathbb R[U_{i\alpha}]\) denote the set of real
polynomials in the matrix entries \(U_{i\alpha}\), and let
\(\tau=\sum_{i,\alpha}U_{i\alpha}\).
Since \(\tau=1\) on the simplex, adding any polynomial multiple
of \(\tau-1\) does not change a polynomial's values there.
The \emph{ideal} \((\tau-1)\) is the set of all such multiples,
namely \((\tau-1)Q\) with \(Q\in R_{n,m}\).
The \emph{quotient} \(\overline R_{n,m}=R_{n,m}/(\tau-1)\)
treats two polynomials as equivalent if their difference belongs
to this ideal. A bar denotes the corresponding equivalence class.
Thus, \(\overline P=\overline Q\) exactly when
\(P-Q=(\tau-1)H\) for some polynomial \(H\).
This is also equivalent to \(P\) and \(Q\) agreeing everywhere
on the simplex, so the quotient identifies precisely those
polynomials that define the same functional on the simplex.
\section{Vanishing at zero rows and collapse for \(n>m\)}
\label{sec:zero-row-vanishing}

We show that boundary regularity imposes a strong restriction:
if \(F\) is \(C^1\)-extendable and satisfies \eqref{eq:dpi}
and \eqref{eq:iz}, then \(F(U)=0\) whenever \(U\) has a zero row.

We begin with a consequence of \eqref{eq:dpi}: merging two
proportional rows preserves the value of \(F\).
For distinct rows \(i,j\), let \(M_{ij}\) be the channel that
merges row \(j\) into row \(i\): it replaces the two rows by
\(U_{i\bullet}+U_{j\bullet}\) and zero, respectively, and fixes
the other rows.

\begin{lemma}[Reversible merging]
\label{lem:parallel-merge-split}
If \(F\) satisfies \eqref{eq:dpi}, then
\(F(M_{ij}U)\le F(U)\), with equality whenever
\(U_{i\bullet}=av\) and \(U_{j\bullet}=bv\) for
\(a,b\ge0\) and \(v\ge0\).
\end{lemma}

\begin{proof}
The merging operation \(M_{ij}\) is represented by a square
column-stochastic matrix, so \eqref{eq:dpi} gives
\(F(M_{ij}U)\le F(U)\).
Suppose now that \(U_{i\bullet}=av\) and \(U_{j\bullet}=bv\).
If \(a+b=0\), both rows are zero, so \(M_{ij}U=U\) and
equality is immediate. Otherwise, merging replaces rows
\(i\) and \(j\) by \((a+b)v\) and zero, respectively.

Define a reverse channel \(S\) that sends state \(i\) to
state \(i\) with probability \(a/(a+b)\) and to state \(j\)
with probability \(b/(a+b)\), while fixing every other state.
This channel is column-stochastic and splits the merged row
back into \(av\) and \(bv\), leaving all other rows unchanged.
Thus \(SM_{ij}U=U\). Applying \eqref{eq:dpi} once more yields
\(F(U)=F(SM_{ij}U)\le F(M_{ij}U)\), proving equality.
\end{proof}

\begin{figure}[htbp]
\centering
\begin{adjustbox}{max width=\linewidth}
\begin{tikzpicture}[
  >=Stealth,
  flow/.style={
    -{Stealth[length=2mm]},
    thin,
    draw=black!55
  },
  note/.style={
    font=\footnotesize,
    text=black!65,
    align=center
  },
  heading/.style={
    anchor=west,
    font=\small\bfseries
  },
  formula/.style={
    font=\small,
    align=center
  }
]
\definecolor{splitblue}{RGB}{38,100,160}
\definecolor{transferorange}{RGB}{180,84,22}

\node[heading] at (0,0)
  {(a) A proportional split and a transfer};

\node[font=\small] at (1.30,-.40) {$U$};
\node[font=\small] at (6.15,-.40) {$U_\eta(0)$};
\node[font=\small] at (12.30,-.40) {$U_\eta(t)$};

\node (U) at (1.30,-1.28) {$\displaystyle
  \begin{pmatrix}
    a&b\\
    c&d\\
    \color{black!40}0&\color{black!40}0
  \end{pmatrix}$};

\node (S) at (6.15,-1.28) {$\displaystyle
  \begin{pmatrix}
    \color{splitblue}(1-\eta)a&
    \color{splitblue}(1-\eta)b\\
    c&d\\
    \color{splitblue}\eta a&
    \color{splitblue}\eta b
  \end{pmatrix}$};

\node (P) at (12.30,-1.28) {$\displaystyle
  \begin{pmatrix}
    (1-\eta)a{\color{transferorange}-t}&(1-\eta)b\\
    c&d\\
    \eta a{\color{transferorange}+t}&\eta b
  \end{pmatrix}$};

\draw[flow] (U.east)--
  node[note,above=3pt] {split row 1}
  (S.west);

\draw[flow] (S.east)--
  node[note,above=3pt] {transfer}
  (P.west);

\draw[thin,draw=black!55]
  (P.south)--(12.30,-2.18);

\draw[flow]
  (12.30,-2.18)--(1.30,-2.18);

\draw[thin,draw=black!55]
  (1.30,-2.18)--(U.south);

\node[note,fill=white,inner xsep=5pt] at (7.2,-2.18)
  {merge rows 1 and 3};

\node[formula] at (7.8,-2.77) {$\displaystyle
  F(U_\eta(t))\ge F(U)=F(U_\eta(0))
  \qquad(-\eta a<t<(1-\eta)a)$};

\node[formula] at (7.8,-3.40) {$\displaystyle
  \left.\frac{\mathrm d}{\mathrm dt}
    F(U_\eta(t))\right|_{t=0}=0
  \quad\Longrightarrow\quad
  \partial_{11}\widetilde F(U_\eta(0))
  =\partial_{31}\widetilde F(U_\eta(0))$};

\draw[black!15] (0,-3.88)--(15.6,-3.88);

\node[heading] at (0,-4.25)
  {(b) The empty row is a common reference};

\node[note] at (7.8,-4.76)
  {Repeat for row 2 and both columns;
   let $\eta\downarrow0$ using $C^1$.};

\node[formula] at (7.8,-5.45) {$\displaystyle
  \underbrace{
    \color{splitblue}\partial_{1\alpha}\widetilde F(U)
  }_{\text{row 1}}
  =
  \underbrace{
    \color{black!65}\partial_{3\alpha}\widetilde F(U)
  }_{\text{empty row}}
  =
  \underbrace{
    \color{transferorange}\partial_{2\alpha}\widetilde F(U)
  }_{\text{row 2}}
  \qquad(\alpha=1,2)$};

\draw[black!15] (0,-6.10)--(15.6,-6.10);

\node[heading] at (0,-6.48)
  {(c) Move row 2 into row 1 without changing the score};

\node (path) at (4.0,-7.55) {$\displaystyle
  U(s)=
  \begin{pmatrix}
    \color{splitblue}a+sc&
    \color{splitblue}b+sd\\
    \color{transferorange}(1-s)c&
    \color{transferorange}(1-s)d\\
    \color{black!40}0&
    \color{black!40}0
  \end{pmatrix}$};

\node (end) at (12.15,-7.55) {$\displaystyle
  U(1)=
  \begin{pmatrix}
    a+c&b+d\\
    0&0\\
    0&0
  \end{pmatrix}$};

\draw[flow] (path.east)--
  node[note,above=3pt] {$0\le s\le1$}
  node[note,below=3pt] {independent endpoint}
  (end.west);

\node[formula] at (7.8,-8.63) {$\displaystyle
  \frac{\mathrm d}{\mathrm ds}F(U(s))
  =
  c\,\underbrace{
    (\partial_{11}-\partial_{21})\widetilde F(U(s))
  }_{0}
  +
  d\,\underbrace{
    (\partial_{12}-\partial_{22})\widetilde F(U(s))
  }_{0}
  =0
  \qquad(0<s<1)$};

\node[formula] at (7.8,-9.42) {$\displaystyle
  \boxed{F(U)=F(U(1))=0}$};

\end{tikzpicture}
\end{adjustbox}

\caption{Zero-row vanishing in the \(3\times2\) case, with
\(a,b,c,d>0\) and \(a+b+c+d=1\).
(a) Proportional splitting preserves \(F\); merging rows 1 and 3 of \(U_{\eta}(t)\) recovers \(U\), so \(t=0\) is a local minimum.
(b) Repeat for both occupied rows and both columns, then let \(\eta\) go to zero. Use the empty row as a bridge to show different rows have the same derivatives. 
(c) Applying the resulting derivative identities along the transfer
path shows that moving all mass from row~2 to row~1 leaves \(F\)
unchanged.}
\label{fig:zero-row-3x2}
\end{figure}

We next show that the empty row provides a common reference:
within each column, the occupied rows have equal partial
derivatives. Mass transfers between them therefore leave
\(F\) unchanged. Moving all mass into one row reaches an
independent distribution, where \(F=0\).
Figure~\ref{fig:zero-row-3x2} illustrates the argument. For the differential argument below, let
\(\widetilde F:\Omega\to\mathbb R\) be a \(C^1\) extension of
\(F\), where \(\Omega\subseteq\mathbb R^{n\times m}\) is open
and contains \(\Delta_{n,m}\).
Every point of the simplex, including its boundary, has an
ambient open ball contained in \(\Omega\). Thus all coordinate
partial derivatives of \(\widetilde F\) are well-defined there;
we write \(\partial_{i\alpha}\widetilde F\) for differentiation
in entry \((i,\alpha)\).
We keep the same extension throughout each proof.
Although coordinate perturbations may leave the simplex,
every application of DPI below uses distributions that remain
in the simplex. The proof uses only within-column mass transfers.
Derivatives along these directions are determined by \(F\)
on the simplex and are independent of the chosen extension.
The role of \(C^1\)-extendability is to ensure that these
derivatives vary continuously up to the boundary.

\begin{lemma}[Coordinate transfers at an unused state]
\label{lem:coordinate-transfer-zero-row}
Suppose \(U\in\Delta_{n,m}\) has zero row \(r\) and is strictly
positive outside that row. If \(F\) is \(C^1\)-extendable and
satisfies \eqref{eq:dpi}, then every \(C^1\) extension
\(\widetilde F\) of \(F\), as in
Definition~\ref{def:C1-functional}, satisfies
\[
 \partial_{i\alpha}\widetilde F(U)
 =\partial_{j\alpha}\widetilde F(U)
\]
for all distinct \(i,j\ne r\) and every \(\alpha\).
\end{lemma}

\begin{proof}
Fix such an extension \(\widetilde F\), a row \(i\ne r\),
and a column \(\alpha\).
For \(0<\eta<1\), split row \(i\) between rows \(i\) and \(r\),
obtaining \(U_\eta\) with rows
\((1-\eta)U_{i\bullet}\) and \(\eta U_{i\bullet}\).
Reversible merging gives \(F(U_\eta)=F(U)\).

Let \(V\in\mathbb R^{n\times m}\) have entries
\(V_{r\alpha}=1\), \(V_{i\alpha}=-1\), and zero elsewhere.
This transfer preserves every column sum, and
\(U_\eta+tV\) remains in the simplex whenever
\(-\eta U_{i\alpha}<t<(1-\eta)U_{i\alpha}\).
Throughout this interval, merging rows \(i\) and \(r\)
recovers \(U\), so
\[
 F(U_\eta+tV)\ge F(U)=F(U_\eta).
\]
Along this feasible path, \(F\) and \(\widetilde F\) agree.
Thus \(t=0\) is also a two-sided local minimum of
\(t\mapsto\widetilde F(U_\eta+tV)\), and the chain rule gives
\[
 0=\left.\frac{\mathrm d}{\mathrm dt}
 \widetilde F(U_\eta+tV)\right|_{t=0}
 =\partial_{r\alpha}\widetilde F(U_\eta)
  -\partial_{i\alpha}\widetilde F(U_\eta).
\]
As \(\eta\downarrow0\), the tables \(U_\eta\) converge to \(U\).
Continuity of the partial derivatives therefore yields
\begin{equation}\label{eq:boundary-transfer}
 \partial_{r\alpha}\widetilde F(U)
 =\lim_{\eta\downarrow0}
   \partial_{r\alpha}\widetilde F(U_\eta)
 =\lim_{\eta\downarrow0}
   \partial_{i\alpha}\widetilde F(U_\eta)
 =\partial_{i\alpha}\widetilde F(U).
\end{equation}
Repeating the argument with row \(j\) gives the same equality
with \(i\) replaced by \(j\). The empty row is thus a common
reference, proving the claim.
\end{proof}

\begin{theorem}[Zero-row vanishing]
\label{thm:zero-row-vanishing}
Let \(n,m\ge2\).
If \(F:\Delta_{n,m}\to\mathbb R\) is \(C^1\)-extendable
and satisfies \eqref{eq:dpi} and \eqref{eq:iz}, then
\(F(U)=0\) whenever \(U\) has a zero row.
\end{theorem}

\begin{proof}
Fix a \(C^1\) extension \(\widetilde F\) of \(F\)
and a zero row \(r\).
First assume that every entry outside row \(r\) is positive.
If \(n=2\), only one row is nonzero, so independence-zero
applies directly. Otherwise, choose \(p\ne r\) and move
all other rows continuously into row \(p\):
\[
 \begin{aligned}
 \Phi_{p\bullet}(s)
   &=U_{p\bullet}+s\sum_{j\ne p,r}U_{j\bullet},\\
 \Phi_{j\bullet}(s)
   &=(1-s)U_{j\bullet}\quad(j\ne p,r),\\
 \Phi_{r\bullet}(s)&=0,
 \end{aligned}
 \qquad 0\le s\le1.
\]
For \(s<1\), all entries outside row \(r\) remain positive.
Since \(F(\Phi(s))=\widetilde F(\Phi(s))\),
Lemma~\ref{lem:coordinate-transfer-zero-row} and the chain
rule imply, for \(0<s<1\),
\[
 \frac{\mathrm d}{\mathrm ds}F(\Phi(s))
 =
 \sum_{j\ne p,r}\sum_{\alpha=1}^m U_{j\alpha}
 \bigl(
   \partial_{p\alpha}\widetilde F(\Phi(s))
   -\partial_{j\alpha}\widetilde F(\Phi(s))
 \bigr)
 =0.
\]
Thus \(F(\Phi(s))\) is constant on \((0,1)\), and by
continuity on \([0,1]\).
The endpoint has only one nonzero row, so
\[
 F(U)=F(\Phi(1))=0.
\]

For an arbitrary \(U\) with row \(r\) zero, let \(B^{(r)}\)
have row \(r\) zero and every other entry equal to
\(1/(m(n-1))\). Apply the preceding argument to
\[
 U_\varepsilon=(1-\varepsilon)U+\varepsilon B^{(r)},
 \qquad 0<\varepsilon<1.
\]
Then \(F(U_\varepsilon)=0\), and continuity gives \(F(U)=0\)
as \(\varepsilon\downarrow0\).
\end{proof}

\begin{remark}[Where mutual information escapes]
\label{rem:shannon-boundary-derivative}
The passage to the boundary in~\eqref{eq:boundary-transfer}
is essential. Let
\(U=\bigl(\begin{smallmatrix}
a&b\\ c&d\\ 0&0
\end{smallmatrix}\bigr)\in\Delta_{3,2}\),
with \(a,b,c,d>0\).
For \(0<\eta<1\), let \(U_\eta\) be obtained by splitting
row 1 into rows 1 and 3 with fractions \(1-\eta\) and \(\eta\).
With
\(V=\bigl(\begin{smallmatrix}
-1&0\\ 0&0\\ 1&0
\end{smallmatrix}\bigr)\),
the same split--merge argument gives
\(\left.\frac{\mathrm d}{\mathrm dt}
I(U_\eta+tV)\right|_{t=0}=0\)
for Shannon mutual information \(I\).
At the boundary, however, direct differentiation gives
\[
 \lim_{t\downarrow0}\frac{\mathrm d}{\mathrm dt}I(U+tV)
 =\log\frac{a+b}{a}>0.
\]
Thus the derivatives along this fixed direction do not
extend continuously to the unused row.
Interior smoothness alone does not yield the zero-row conclusion.
\end{remark}

\subsection{Collapse when the processed alphabet is larger}
\label{sec:nxm}

We now apply the zero-row theorem to the case \(n>m\).
To illustrate the argument, consider a \(3\times2\) distribution
\[
 U=
 \begin{pmatrix}
 a&b\\
 c&d\\
 e&f
 \end{pmatrix}.
\]
Every such \(U\) can be written as \(TU^\flat\), where \(T\)
is a column-stochastic matrix and
\[
 U^\flat=
 \begin{pmatrix}
 a+c+e&0\\
 0&b+d+f\\
 0&0
 \end{pmatrix},
\]
which has the same column marginal and a zero row.
The zero-row theorem gives \(F(U^\flat)=0\), and
\eqref{eq:dpi} and nonnegativity then imply
\(0\le F(U)\le F(U^\flat)=0\).
The same argument applies whenever \(n>m\).

\begin{theorem}[Tall-case collapse]
\label{thm:normalized}
Let \(n>m\ge2\).
Every \(C^1\)-extendable functional on \(\Delta_{n,m}\)
satisfying \eqref{eq:dpi} and \eqref{eq:iz}
is identically zero.
\end{theorem}

\begin{proof}
For \(U\in\Delta_{n,m}\), put
\(s_\alpha=\sum_iU_{i\alpha}\), and define
\[
 U^\flat_{i\alpha}=
 \begin{cases}
 s_\alpha,&i=\alpha,\\
 0,&i\ne\alpha.
 \end{cases}
\]
The last \(n-m\) rows are zero, so
Theorem~\ref{thm:zero-row-vanishing} gives \(F(U^\flat)=0\).

Construct a column-stochastic \(T\) whose \(\alpha\)-th column
is \(U_{\bullet\alpha}/s_\alpha\) whenever \(s_\alpha>0\).
Choose its other columns arbitrarily as probability vectors.
A column of \(U\) with \(s_\alpha=0\) is zero, so
\(TU^\flat=U\).

By Lemma~\ref{lem:refPone-derived}, \eqref{eq:dpi}
and \eqref{eq:iz} give \(F(U)\ge0\).
Applying \eqref{eq:dpi} to \(TU^\flat=U\), we therefore obtain
\[
 0\le F(U)\le F(U^\flat)=0.
\]
\end{proof}

\begin{corollary}[Polynomial form]
\label{cor:normalized-polynomial}
If \(n>m\ge2\) and \(P\in R_{n,m}\) satisfies
\eqref{eq:dpi} and \eqref{eq:iz} on \(\Delta_{n,m}\),
then \(P\in(\tau-1)\).
\end{corollary}

\begin{proof}
A polynomial is \(C^1\)-extendable, so
Theorem~\ref{thm:normalized} gives vanishing on the simplex.
Eliminating one entry by \(\tau=1\) yields a polynomial
that vanishes on a nonempty open set in the other \(nm-1\)
entries. It is identically zero, which is exactly
\(P\in(\tau-1)\).
\end{proof}

\paragraph{AI-assistance}
The zero-row argument was developed by the author, with aspects
of the approach inspired by earlier interactions with GPT-5.5.
GPT-5.6 subsequently assisted in simplifying arguments.
\section{Determinantal rigidity and the sharp degree bound}
\label{sec:non-tall}

Assume \(2\le n\le m\). For \(S\subseteq[m]\) with \(|S|=n\),
write \(\Delta_S(U)=\det U_{\bullet S}\), and let
\[
 I_n=\langle\Delta_S:|S|=n\rangle\subseteq R_{n,m},
 \qquad
 \overline I_n=(I_n+(\tau-1))/(\tau-1).
\]
The common zero set of the maximal minors is the set of matrices with
rank less than \(n\). We show that every admissible polynomial belongs
to the square of this ideal after imposing the simplex equation.

\begin{theorem}[Determinantal rigidity]\label{thm:non-tall}
Let \(2\le n\le m\), and let \(P\in R_{n,m}\) satisfy \eqref{eq:dpi} and \eqref{eq:iz} on \(\Delta_{n,m}\). Then
\[
 P\in I_n^2+(\tau-1),
 \qquad\text{equivalently}\qquad
 \overline P\in\overline I_n^{\,2}.
\]
If \(P\) is not identically zero on the simplex, then
\(\deg P\ge2n\). The bound is attained by
\(F_{\mathrm{Gram}}(U)=\det(UU^\top)\).
\end{theorem}

\subsection{From zero rows to an open patch of rank deficiency}

\begin{lemma}[Row-cone vanishing]\label{lem:row-cone-vanishing}
If \(F\) is \(C^1\)-extendable and satisfies \eqref{eq:dpi} and \eqref{eq:iz},
then \(F(U)=0\) whenever one row of \(U\in\Delta_{n,m}\) is a
nonnegative linear combination of the others.
\end{lemma}
\begin{proof}
By permutation invariance, assume
\(U_{n\bullet}=\sum_{i<n}\lambda_iU_{i\bullet}\), with
\(\lambda_i\ge0\). Set
\(U'_{i\bullet}=(1+\lambda_i)U_{i\bullet}\) for \(i<n\)
and \(U'_{n\bullet}=0\). Then \(U'\in\Delta_{n,m}\), and
zero-row vanishing gives \(F(U')=0\). Define
{\small
\[
T=
\begin{pmatrix}
\frac1{1+\lambda_1}&0&\cdots&0&0\\
0&\frac1{1+\lambda_2}&\cdots&0&0\\
\vdots&\vdots&\ddots&\vdots&\vdots\\
0&0&\cdots&\frac1{1+\lambda_{n-1}}&0\\
\frac{\lambda_1}{1+\lambda_1}
&\frac{\lambda_2}{1+\lambda_2}
&\cdots&\frac{\lambda_{n-1}}{1+\lambda_{n-1}}&1
\end{pmatrix}.
\]
}
Each column is a probability vector, so \(T\) is
column-stochastic. Direct multiplication gives
\(TU'=U\). Hence nonnegativity and \eqref{eq:dpi} yield
\[
0\le F(U)=F(TU')\le F(U')=0.
\]
\end{proof}

\paragraph{Homogenization.}
To prove Theorem~\ref{thm:non-tall}, we focus on nonzero polynomials. For any nonzero \(P\), write \(\kappa=\deg P\) and
\(P=\sum_{k=0}^{\kappa}P_k\), where \(P_k\) is homogeneous of
degree \(k\). Define
\begin{equation}\label{eq:homogenization}
 H=\sum_{k=0}^{\kappa}\tau^{\kappa-k}P_k.
\end{equation}
Then \(H-P\in(\tau-1)\), and for \(s=\tau(U)>0\),
\[
 H(U)=s^{\kappa}P(U/s).
\]
It follows that \(H\ge0\) on the nonnegative cone of positive total
mass. By Lemma~\ref{lem:row-cone-vanishing}, it is zero there whenever
one row is a nonnegative combination of the others. Each summand \(\tau^{\kappa-k}P_k\) is homogeneous of
total degree \(\kappa\), so \(H\) is either zero or
homogeneous of degree \(\kappa\). If \(P\) is nonzero
as a functional on the simplex, then \(H=P\) there
implies \(H\ne0\) and \(\deg H=\kappa\).

The nonnegativity and vanishing properties of \(P\) are initially
known only on the simplex, where \(\tau=1\).
Homogenization gives a polynomial \(H\) with these properties
on the nonnegative cone. This allows us to perturb individual
matrix entries near a positive matrix without preserving total
mass, as needed in the argument below.

\subsection{Second-order vanishing in a Schur chart}

Put \(r=m-n+1\), and partition the generic matrix as
\[
 U=\begin{pmatrix}A&B\\c^\top&d^\top\end{pmatrix},
 \qquad
 A\in\mathbb R^{(n-1)\times(n-1)},\quad
 B\in\mathbb R^{(n-1)\times r}.
\]

When \(A\) is invertible, the Schur coordinates
\[
 q^\top=d^\top-c^\top A^{-1}B
\]
measure the last row's failure to lie in the span of the first \(n-1\)
rows. We have \( q=0\Leftrightarrow \rank(U)< n\).

Let \(\delta=\det A\). We consider \((R_{n,m})_\delta=\{f\delta^{-k}|f\in R_{n,m},k=0,1,2,\cdots\}\), the ring allowing powers of
\(\delta\) in denominators. Since \(\delta=\det A\) is invertible in \((R_{n,m})_\delta\), we have \(
A^{-1}=\delta^{-1}\operatorname{adj}(A),\) where \(\operatorname{adj}(A)\) denotes the adjugate matrix of \(A\). Its entries are polynomials in the entries of \(A\), so all entries of \(A^{-1}\) belong to \((R_{n,m})_\delta\).

If \(S\) is the polynomial ring in the entries of \(A,B,c\),
then \(R_{n,m}=S[d_1,\ldots,d_r]\), and the invertible change of
coordinates \(d^\top=q^\top+c^\top A^{-1}B\) gives
\[
 (R_{n,m})_\delta=S_\delta[q_1,\ldots,q_r].
\]
Write
\((I_n)_\delta
:=\{f\delta^{-k}:f\in I_n,\ k=0,1,2,\ldots\}
\)
for the ideal generated by the maximal minors in the localized ring
\((R_{n,m})_\delta\). Since the entries of \(c^\top A^{-1}\) belong to
\((R_{n,m})_\delta\), we may perform the row operation
\[
U_{n\bullet}\longmapsto
U_{n\bullet}-\sum_{i=1}^{n-1}(c^\top A^{-1})_iU_{i\bullet},
\]
which transforms \(U\) into
\[
\widetilde U=
\begin{pmatrix}
A&B\\
0&q^\top
\end{pmatrix}
\]
and leaves every maximal minor unchanged. Expanding each maximal
minor of \(\widetilde U\) along its last row shows that it belongs
to \((q_1,\ldots,q_r)\). Hence
\[
(I_n)_\delta\subseteq(q_1,\ldots,q_r).
\]
Conversely, for each \(\beta=1,\ldots,r\), the minor using the first
\(n-1\) columns and column \(n-1+\beta\) equals \(\delta q_\beta\).
Since \(\delta^{-1}\) is in \((R_{n,m})_\delta\) and \(\delta q_\beta\) is in \((I_n)_\delta\), we have
\(q_\beta=\delta^{-1}(\delta q_\beta)\in(I_n)_\delta\).
Thus the reverse inclusion also holds, and therefore
\begin{equation}\label{eq:schur-ideal}
(I_n)_\delta=(q_1,\ldots,q_r).
\end{equation}

We next exhibit a nonempty open set of positive row-cone points
in this chart. With \(\ones_m=(1,\ldots,1)^\top\in\mathbb R^{m\times1}\),
standard basis vectors \(e_i\in\mathbb R^{m\times1}\), and
\(\varepsilon>0\), take the first \(n-1\) rows to be
\(
v_i=(\ones_m+\varepsilon e_i)^{\top}\), \(i=1,\ldots,n-1,
\)
and the last row to be \(v_n=\sum_{i<n}v_i\).
Denote the corresponding blocks by \(A_0,B_0,c_0,d_0\), and put
\(z_0=(A_0,B_0,c_0)\). All entries of this matrix are strictly
positive. Moreover, \(A_0=J+\varepsilon I\), where \(J\) is the
\((n-1)\times(n-1)\) all-ones matrix, so \(A_0\) is invertible.
The identity \(v_n=\sum_{i<n}v_i\) gives
\[
c_0^\top=\mathbf1_{n-1}^{\top}A_0,\quad
d_0^\top=\mathbf1_{n-1}^{\top}B_0,\quad
c_0^\top A_0^{-1}=\mathbf1_{n-1}^{\top}.
\]
Consequently,
\[
d_0^\top-c_0^\top A_0^{-1}B_0=0.
\]

For \(z=(A,B,c)\) with \(\det A\ne0\), define
\(
\lambda(z)^\top:=c^\top A^{-1},
\)
and take
\[
O:=\left\{
z=(A,B,c):
\det A\ne0,\quad A>0,\quad B>0,\quad \lambda(z)>0
\right\},
\]
where all inequalities are interpreted entrywise.
The map \(z\mapsto\lambda(z)\) is continuous on
\(\{\det A\ne0\}\), so \(O\) is an open neighborhood of \(z_0\)
in the Euclidean space of coordinates \((A,B,c)\). For every \(z\in O\), setting \(q=0\) gives
\[
d^\top=c^\top A^{-1}B=\lambda(z)^\top B,
\qquad
(c^\top,d^\top)=\lambda(z)^\top(A\ B).
\]
Thus the matrix with Schur coordinates \((z,0)\) has strictly
positive entries, and its last row is a positive linear
combination of the first \(n-1\) rows.

Expand \(H\) in the \(q\)-variables over \(S_\delta\):
\[
H(z,q)=C_0(z)+\sum_{\beta=1}^r C_\beta(z)q_\beta
       +\sum_{|a|\ge2}C_a(z)q^a.
\]
For each fixed \(z\in O\), row-cone vanishing gives \(H(z,0)=0\).
Since the matrix at \((z,0)\) has strictly positive entries,
all sufficiently small perturbations of \(q\) near zero, in either
direction, preserve positivity. On these matrices \(H\ge0\),
so \(q\mapsto H(z,q)\) has a local minimum at \(q=0\).
It follows that
\[
C_0(z)=H(z,0)=0,
\qquad
C_\beta(z)
=\frac{\partial H}{\partial q_\beta}(z,0)=0,
\quad \beta=1,\ldots,r.
\]

Each of \(C_0,C_1,\ldots,C_r\) belongs to \(S_\delta\).
After clearing its denominator, its numerator is a polynomial
in \(z\) that vanishes on the nonempty Euclidean open set \(O\),
and hence is identically zero. Therefore
\[
H\in(q_1,\ldots,q_r)^2
  =\bigl((I_n)_\delta\bigr)^2
  =(I_n^2)_\delta.
\]
Equivalently, there exists \(N\ge0\) such that
\begin{equation}\label{eq:local-membership}
\delta^N H\in I_n^2.
\end{equation}

\subsection{Removing the pivot factor}\label{sec:remove-pivot}

\paragraph{Square matrices.}
If \(m=n\), then \(I_n=(D)\), where \(D=\det U\).
The generic determinant is irreducible over \(\mathbb R\)
(Lemma~\ref{thm:generic-det-irreducible}). Since \(R_{n,n}\) is a
unique factorization domain, \(D\) is prime. Moreover, \(D\nmid\delta\),
because their degrees are \(n\) and \(n-1\). Thus
\(D^2\mid\delta^NH\) in~\eqref{eq:local-membership} implies
\(D^2\mid H\), or \(H\in I_n^2\).

\paragraph{Rectangular matrices.}
For \(m>n\), we use the following standard consequence of
determinantal algebra.
\begin{proposition}[Maximal-minor saturation~{\cite[Corollary~9.18]{BrunsVetter}}]
\label{prop:maximal-minor-square-saturation}
The generic maximal-minor ideal \(I_n\) is prime, and \(I_n^2\) is
\(I_n\)-primary. In particular, if \(a\notin I_n\) and
\(af\in I_n^2\), then \(f\in I_n^2\).
\end{proposition}
The final implication is the defining cancellation property of a
primary ideal with radical \(I_n\). Since \(I_n\) is homogeneous
and generated in degree \(n\), the degree-\((n-1)\) polynomial
\(\delta\) does not belong to \(I_n\). Primality gives
\(\delta^N\notin I_n\), and the proposition applied
to~\eqref{eq:local-membership} yields \(H\in I_n^2\).

In both cases, \(H\in I_n^2\), and \(H-P\in(\tau-1)\) proves
the structural results of Theorem~\ref{thm:non-tall}.
If \(P\) is nonzero on the simplex, \(H\) is a nonzero homogeneous
polynomial of degree \(\deg P\). Every nonzero element of \(I_n^2\)
has degree at least \(2n\), because that ideal is generated by
homogeneous polynomials of degree \(2n\). This proves the degree
bound. Notice that the argument applies before quotienting by
\(\tau-1\); the simplex relation cannot lower the required sample
count.

\subsection{Sharpness and the square specialization}

\begin{proposition}[Gram-determinant monotone]\label{prop:gram-determinant}
For \(2\le n\le m\), the polynomial
\(F_{\mathrm{Gram}}(U)=\det(UU^\top)\) is nonzero, has degree
\(2n\), and satisfies \eqref{eq:dpi} and \eqref{eq:iz}.
\end{proposition}
\begin{proof}
Cauchy--Binet gives
\[
 F_{\mathrm{Gram}}(U)=\sum_{\substack{S\subseteq[m]\\|S|=n}}
                         \det(U_{\bullet S})^2.
\]
It vanishes whenever \(\rank U<n\), in particular at independence.
It is positive on a matrix with positive mass at \(n\) distinct
diagonal positions, so it is nonzero and has degree \(2n\).
For a column-stochastic \(T\),
\[
 F_{\mathrm{Gram}}(TU)=(\det T)^2 F_{\mathrm{Gram}}(U).
\]
Each column of \(T\) has Euclidean norm at most one. Hadamard's
inequality gives \(|\det T|\le1\), proving fixed-alphabet DPI.
\end{proof}

This completes the proof of Theorem~\ref{thm:non-tall}, including
attainment of the degree bound. In the square case it gives the
following especially simple form.

\begin{corollary}[Square-case divisibility]\label{thm:nxn}\label{sec:nxn}
Let \(n\ge2\). Every polynomial \(P\) satisfying \eqref{eq:dpi} and
\eqref{eq:iz} on \(\Delta_{n,n}\) has the form
\[
 P(U)=(\det U)^2G(U)\qquad(U\in\Delta_{n,n})
\]
for some polynomial \(G\). A nonzero such functional requires at
least \(2n\) samples for unbiased estimation, and \((\det U)^2\)
attains this bound.
\end{corollary}
\begin{proof}
Here \(I_n=(\det U)\), so the representation follows from
Theorem~\ref{thm:non-tall}. Proposition~\ref{prop:finite-sample-polynomial}
translates its degree bound to a sample bound, and
\(\det(UU^\top)=(\det U)^2\).
\end{proof}

The proof also shows that, at the minimum degree \(2n\), the
square-case functional is a positive scalar multiple of \((\det U)^2\). It provides a necessary form. Appendix~\ref{app:constant-compensation-full-dpi} shows that there exists a large family of $G(U)$ that makes \((\det U)^2G(U)\) satisfy \eqref{eq:dpi} and \eqref{eq:iz}. Appendix~\ref{app:hardness} shows that recognizing whether a polynomial satisfies DPI is coNP-hard.

\paragraph{AI-assistance} GPT-5.5 provided substantive assistance in developing the determinantal-rigidity proofs within a proof-strategy
framework supplied by the author.
\section{Consequences for information measures}\label{sec:consequences}

\subsection{Faithful detection of independence}\label{sec:faithfulness}

A functional is \emph{independence-faithful} if
\(F(U)=0\) exactly when \(\rank U=1\). The zero-row theorem rules
out this stronger property as soon as the processed alphabet has at
least three states.

\begin{theorem}[Faithfulness exactly when \(n=2\)]
\label{thm:C1-one-sided-faithfulness}
For \(n,m\ge2\), there exists an independence-faithful,
\(C^1\)-extendable functional on \(\Delta_{n,m}\) satisfying
\eqref{eq:dpi} and \eqref{eq:iz} if and only if \(n=2\).
When \(n=2\), the Gram determinant is such a functional.
\end{theorem}
\begin{proof}
If \(n\ge3\), the distribution with
\(U_{11}=U_{22}=1/2\) and all other entries zero has a zero row
and rank two. The zero-row theorem forces value zero, contradicting
faithfulness.

If \(n=2\), Cauchy--Binet gives
\[
 \det(UU^\top)
 =\sum_{\alpha<\beta}
   (U_{1\alpha}U_{2\beta}-U_{1\beta}U_{2\alpha})^2.
\]
This vanishes exactly when \(\rank U<2\), which on the simplex
means \(\rank U=1\). Its \eqref{eq:dpi} follows from
Proposition~\ref{prop:gram-determinant}.
\end{proof}

\(F\) satisfies \emph{two-sided fixed-alphabet DPI} if
\begin{equation}
    F(TUS^\top)\le F(U)
    \tag{TF-DPI}\label{eq:two-sided-fixed-alphabet-dpi}
\end{equation}
for every \(U\in\Delta_{n,m}\) and every pair of column-stochastic
matrices \(T\in\mathbb R^{n\times n}\) and
\(S\in\mathbb R^{m\times m}\).

\begin{corollary}[Faithfulness with fixed-sample unbiased estimation]
\label{thm:one-sided-faithfulness}\label{cor:two-sided-impossibility-triangle}
An independence-faithful, finitely unbiasedly estimable functional
satisfying \eqref{eq:dpi} exists exactly when \(n=2\). Its minimum
sample count is four. Under \eqref{eq:two-sided-fixed-alphabet-dpi}, such a functional exists
exactly when \(n=m=2\), again with minimum sample count four.
\end{corollary}
\begin{proof}
Finite-sample unbiased estimability implies polynomiality and hence
\(C^1\)-extendability. Necessity follows from the theorem; under
two-sided DPI apply it also to \(U\mapsto F(U^\top)\). For \(n=2\),
the Gram determinant has degree four, and
Theorem~\ref{thm:non-tall} rules out smaller degree. When \(n=m=2\),
it equals \((\det U)^2\), which satisfies fixed-alphabet DPI on both sides.
\end{proof}

\paragraph{Without the independence-zero normalization.}
\label{par:one-sided-normalization}
The preceding results extend to monotones without
independence-zero by subtracting the appropriate baseline.

\begin{itemize}
\item \emph{One-sided DPI.}
Under \eqref{eq:dpi}, define
\(\varphi(b):=F(e_1b^\top)\), where \(b\) is a probability
vector in \(\mathbb R^m\) and \(e_1\) is the first standard
basis vector of \(\mathbb R^n\).
Then \(F_0(U):=F(U)-\varphi(c(U))\) satisfies
\eqref{eq:dpi} and \eqref{eq:iz}, with the same
\(C^1\)-extendability or polynomiality as \(F\).
Consequently, if \(F\) is \(C^1\)-extendable and \(n>m\),
then \(F(U)=\varphi(c(U))\).
If \(F\) is polynomial and \(n\le m\), then
\(F_0\in I_n^2+(\tau-1)\); in particular, any nonzero
\(F_0\) requires degree and unbiased sample count at least
\(2n\).

\item \emph{Two-sided DPI.}
Under \eqref{eq:two-sided-fixed-alphabet-dpi}, all independent
distributions have a common value \(\gamma\), and
\(F-\gamma\) satisfies the same DPI and independence-zero.
Thus every \(C^1\)-extendable such \(F\) is constant when
\(n\ne m\). When \(n=m\) and \(F\) is polynomial,
\(F=\gamma+(\det U)^2G\) on the simplex.
Any nonconstant such functional requires degree and
unbiased sample count at least \(2n\).
\end{itemize}

Appendix~\ref{sec:normalization} gives the formal statements
and proofs, together with consequences for independence tests.

\subsection{Allowing the alphabet size to change}\label{sec:expansibility}

Consider a family
\(\{F_{n,m}:\Delta_{n,m}\to\mathbb R\}_{n,m\ge2}\).
It satisfies left-sided dimension-changing DPI if
\begin{equation}
    F_{k,m}(TU)\le F_{n,m}(U)
    \tag{LC-DPI}\label{eq:dimension-changing-dpi}
\end{equation}
for every \(k,n,m\ge2\), every \(k\times n\) column-stochastic
matrix \(T\), and every \(U\in\Delta_{n,m}\).

Adding a zero row is reversible by a stochastic channel: the embedding
maps \(U\) to \(U^+=\bigl(\begin{smallmatrix}U\\0\end{smallmatrix}\bigr)\),
and merging the added state into an original state maps \(U^+\) back
to \(U\). Dimension-changing DPI therefore implies
\begin{equation}\label{eq:expansibility}
 F_{n+1,m}(U^+)=F_{n,m}(U).
\end{equation}
This property is often called \emph{expansibility}.

\begin{corollary}[Dimension-changing collapse]
\label{cor:expansibility-collapse}
If a family satisfies \eqref{eq:dimension-changing-dpi} and every
\(F_{n,m}\) is \(C^1\)-extendable and satisfies \eqref{eq:iz}, then
\(F_{n,m}\equiv0\) for all \(n,m\ge2\).
\end{corollary}
\begin{proof}
The zero-row theorem gives \(F_{n+1,m}(U^+)=0\).
Equation~\eqref{eq:expansibility} gives \(F_{n,m}(U)=0\).
\end{proof}

Fixed-alphabet
DPI leaves room for polynomial monotones. However, dimension-changing DPI forces collapse for all \(C^1\)-extendable functions, including polynomials.

\section{Sharp task bounds for multi-task peer prediction}\label{sec:peer-prediction}

In multi-task peer prediction, a principal, such as an online
platform, wants to elicit honest judgments from agents on several
similar tasks without access to ground truth. For example, a
platform may ask two agents, Alice and Bob, to rate the same
collection of movies. For each movie, each agent forms a private
judgment, called her \emph{signal}, and is asked to report it
honestly. Their signal spaces need not be the same: Alice's
signal might indicate whether she liked the movie, while Bob's
might be a rating from one to five stars. The platform observes
their reports but cannot directly verify their private judgments.
A mechanism specifies each agent's payment based jointly on
both agents' reports across the tasks, with the aim of making
truthful reporting optimal.

We model the agents' signal pairs across tasks as independent
draws from the same unknown joint distribution. Under the
task-by-task reporting strategies specified below, the report
pairs are also independent and identically distributed.
For any payment rule using \(\ell\) tasks, its expected payment
is therefore a polynomial of degree at most \(\ell\) in the
entries of the joint report distribution: the probability of
each possible sequence of reports is a product of \(\ell\)
such entries. This connects finite-task payments to polynomial
dependence measures and makes our structural results applicable
to every finite-task payment rule in the following model.

\paragraph{Signals, reports, and strategies.}
Alice and Bob receive \(\ell\) signal pairs
\((X_t,Y_t)\in[n]\times[m]\), drawn i.i.d.\ from an unknown
\(P\in\Delta_{n,m}\), and report in the same alphabets.
A mechanism consists of payment functions
\[
 \mathcal M_A,\mathcal M_B:([n]\times[m])^\ell\to\mathbb R.
\]
The rule is fixed independently of \(P\), and its incentive conditions
are required for every \(P\). Payments may be signed.

We consider \emph{memoryless strategies}: Alice applies the same
\(n\times n\) column-stochastic channel \(T\) independently on
every task, and Bob applies an \(m\times m\) channel \(S\) in the
same way. The reported pairs are then i.i.d.\ from \(TPS^\top\).
Truthful reporting is the identity channel. A mechanism is
\emph{dominantly truthful in this strategy model} if, against every
fixed channel of the other agent, truthful reporting maximizes one's
expected payment over all own channels. This definition does not
assert truthfulness against arbitrary joint strategies on the entire
vector of signals.

\paragraph{Normalization and nontriviality.}
An agent's payment is \emph{independence-zero} if its expectation is
zero under every i.i.d.\ joint report distribution of rank one.
It is \emph{nontrivial} if its truthful expected payment is nonzero
for at least one \(P\). Nontriviality refers to the expected payment;
a realized payment rule can be nonzero and yet have zero expectation
under every distribution.

\begin{lemma}[Expected payments are polynomial monotones]
\label{lem:mech-to-measure}
Suppose Alice's payment is independence-zero and she cannot gain by
any memoryless deviation when Bob is truthful, for any \(P\).
Then
\(F_A(U)=\mathbb E_{U^{\otimes\ell}}[\mathcal M_A]\) is represented
by a polynomial of degree at most \(\ell\) satisfying \eqref{eq:dpi}
and \eqref{eq:iz}.
\end{lemma}
\begin{proof}
The expected payment is
\begin{equation}\label{eq:expected-payment}
 F_A(U)=\sum_{((x_t,y_t))_{t=1}^{\ell}}
          \mathcal M_A\bigl((x_t,y_t)_{t=1}^{\ell}\bigr)
          \prod_{t=1}^{\ell}U_{x_ty_t},
\end{equation}
where the sum is over \(([n]\times[m])^\ell\). This is a polynomial
of degree at most \(\ell\). Holding Bob truthful, a channel \(T\)
changes the report distribution from \(U\) to \(TU\). The incentive
condition is therefore exactly \(F_A(TU)\le F_A(U)\). The independence-zero requirement for the payment rule corresponds to the independence-zero condition of \(F_A(U)\). 
\end{proof}

\begin{theorem}[Optimal task bounds]\label{thm:task-bounds}
Consider mechanisms in the model above with independence-zero
expected payments.
\begin{enumerate}[label=(\roman*)]
\item\label{item:tall-payment}
If \(n>m\), every finite-task mechanism that is truthful for Alice
against all memoryless deviations with Bob truthful has
\(F_A\equiv0\). Consequently, for unequal alphabet sizes, no
finite-task dominantly truthful mechanism has nontrivial expected
payments for both agents.
\item\label{item:square-payment}
If \(n=m\), any mechanism that is truthful for Alice against all memoryless deviations with Bob truthful, and has nontrivial Alice payment, requires
\(\ell\ge2n\). There is a \(2n\)-task mechanism with nontrivial
payments for both agents that is dominantly truthful in the
memoryless strategy model.
\end{enumerate}
\end{theorem}
\begin{proof}
By Lemma~\ref{lem:mech-to-measure}, \(F_A\) is a polynomial
satisfying \eqref{eq:dpi} and \eqref{eq:iz}. If \(n>m\), tall-case collapse
gives \(F_A\equiv0\). When \(m>n\), the same argument applies to
Bob after transposing the joint distribution. This proves (i).

For \(n=m\), nontriviality and Corollary~\ref{thm:nxn} imply
\(\ell\ge\deg F_A\ge2n\). For attainment, pay both agents an
unbiased \(2n\)-sample estimator of \((\det U)^2\), which exists
by Proposition~\ref{prop:finite-sample-polynomial}; this is the
expected-payment construction of DMI~\cite{DBLP:conf/soda/Kong20}.
Against any fixed Bob channel \(S\), put \(V=PS^\top\). Alice's
channel \(T\) gives expected payment
\((\det(TV))^2\le(\det V)^2\). The same argument applies to Bob,
since the squared determinant satisfies fixed-alphabet DPI on both sides. A
full-rank distribution gives positive truthful expected payments.
\end{proof}

\section{Four-task peer prediction under restricted deviations}
\label{sec:restricted}

By requiring truthfulness only against a restricted class of
reporting strategies, we obtain a four-task peer-prediction
mechanism with weaker incentive guarantees.

\subsection{A classical dependence measure as a payment target}

Consider the squared Euclidean distance between the joint distribution
and the product of its marginals:
\begin{equation}\label{eq:QMI}
 \QMI(U):=\|U-r(U)c(U)^\top\|_F^2
 =\sum_{i,\alpha}\bigl(U_{i\alpha}-r_i(U)c_\alpha(U)\bigr)^2.
\end{equation}
This functional is called \emph{Quadratic Mutual Information} (QMI) and has been proposed previously~\cite{Torkkola2003, BerrettSamworth2021USP}. 

Since the marginals are linear in \(U\), \(\QMI\) is represented by a
polynomial of degree at most four. Moreover,
\[
 \QMI(U)\ge0,\qquad
 \QMI(U)=0\ \Longleftrightarrow\
 U=r(U)c(U)^\top\ \Longleftrightarrow\ \rank U=1.
\]
Our use of this established functional is as an expected reward. The
question is which reporting strategies it discourages, uniformly over
the unknown signal distribution.

\subsection{Relabelings and signal-independent reports}

Let \(\mathsf{DS}_n\) be the doubly stochastic \(n\times n\) matrices
and let
\[
 \mathsf{R1}_n=\{q\ones_n^\top:q\ge0,\ \ones_n^\top q=1\},
 \qquad
 \mathsf C_n=\operatorname{conv}(\mathsf{DS}_n\cup\mathsf{R1}_n).
\]
By the Birkhoff--von Neumann theorem, a doubly stochastic channel is a
mixture of permutations: on each task the agent can randomly relabel
her signal. A rank-one channel discards the signal and generates a
report from a fixed distribution. Every channel in \(\mathsf C_n\)
can be written as
\[
 T=\lambda D+(1-\lambda)q\ones_n^\top,\qquad
 0\le\lambda\le1,\quad D\in\mathsf{DS}_n.
\]
It independently chooses, on each task, between these two operations.
The same resulting channel is used on all tasks, as in the memoryless
model of Section~\ref{sec:peer-prediction}.

\begin{proposition}[Restricted monotonicity]\label{prop:frob-dpi}
For every \(U\in\Delta_{n,m}\) and \(T\in\mathsf C_n\) represented
as above,
\[
 \QMI(TU)\le\lambda^2\QMI(U)\le\QMI(U).
\]
The analogous inequality holds for right processing by
\(S\in\mathsf C_m\).
\end{proposition}
\begin{proof}
Write \(M(U)=U-r(U)c(U)^\top\). Column stochasticity gives
\(r(TU)=Tr(U)\), \(c(TU)=c(U)\), and hence
\(M(TU)=TM(U)\). Since \(\ones_n^\top M(U)=0\), the rank-one
part of \(T\) annihilates \(M(U)\), so
\[
 M(TU)=\lambda DM(U).
\]
A convex combination of permutation matrices is a contraction in
Euclidean norm. Apply this to each column to obtain
\(\|DM(U)\|_F\le\|M(U)\|_F\), proving the claim.
Right processing follows from
\(\QMI(U^\top)=\QMI(U)\).
\end{proof}

This is a restricted incentive condition. It includes all relabelings
and all signal-independent reports, but does not include arbitrary
report strategies.

\subsection{The four-task payment and its incentive guarantee}\label{subsec:four-task}

For four reported pairs \(z_t=(x_t,y_t)\), define the common payment
\begin{equation}\label{eq:QMI-payment}
 \MQMI(z_1,z_2,z_3,z_4)
 =\ind\{x_1=x_2,\ y_1=y_2\}
 -2\ind\{x_1=x_2,\ y_1=y_3\}
 +\ind\{x_1=x_2,\ y_3=y_4\}.
\end{equation}
Both agents receive this payment. It uses only equality comparisons of
reports, needs no estimate of the unknown distribution, and may be
negative on an individual realization. This payment rule is an unbiased estimator of QMI.

\begin{lemma}[Four-sample unbiased estimation for QMI]
\label{prop:fmi}
For i.i.d.\ \(Z_1,\ldots,Z_4\) with joint distribution \(U\),
\[
 \mathbb E_{U^{\otimes4}}[\MQMI]=\QMI(U).
\]
\end{lemma}
\begin{proof}
The expectations of the three indicator terms in
\eqref{eq:QMI-payment} are, respectively,
\[
 \sum_{i,\alpha}U_{i\alpha}^2,\qquad
 \sum_{i,\alpha}U_{i\alpha}r_i(U)c_\alpha(U),\qquad
 \sum_{i,\alpha}r_i(U)^2c_\alpha(U)^2.
\]
Their prescribed linear combination is \eqref{eq:QMI}.
\end{proof}

\begin{theorem}[Four-task rewards under restricted deviations]
\label{cor:four-tasks}
For every \(n,m\ge2\), the payment rule \eqref{eq:QMI-payment} has
the following properties in the memoryless strategy model.
\begin{enumerate}[label=(\roman*)]
\item For every signal distribution \(P\) and every fixed Bob channel
\(S\), Alice's truthful expected payment is at least her expected
payment under any deviation \(T\in\mathsf C_n\).
\item Symmetrically, for every fixed Alice channel, Bob cannot gain by
a deviation in \(\mathsf C_m\).
\item When both agents are truthful, their common expected payment is
zero exactly at independence and is positive otherwise.
\end{enumerate}
Thus truthful reporting is a dominant strategy relative to the
restricted own-deviation classes, even against unrestricted fixed
memoryless strategies of the peer.
\end{theorem}
\begin{proof}
Fix \(P\) and a Bob channel \(S\), and put \(V=PS^\top\).
Lemma~\ref{prop:fmi} identifies Alice's truthful expected payment as
\(\QMI(V)\), and her payment after deviation \(T\) as
\(\QMI(TV)\). Proposition~\ref{prop:frob-dpi} proves (i).
Transpose to obtain (ii). With \(T=S=I\), the expectation is
\(\QMI(P)\), giving (iii).
\end{proof}

\begin{center}
\renewcommand{\arraystretch}{1.2}
\begin{tabular}{@{}p{0.30\linewidth}p{0.29\linewidth}p{0.29\linewidth}@{}}
\toprule
Own deviations & Equal alphabets \(n=m\) & Unequal alphabets\\
\midrule
All stochastic channels
& \(2n\) tasks necessary and sufficient for nontrivial rewards
& Larger-alphabet expected reward is zero\\
Channels in \(\mathsf C_n,\mathsf C_m\)
& Four tasks suffice; reward positive at every dependent \(P\)
& Four tasks suffice; reward positive at every dependent \(P\)\\
\bottomrule
\end{tabular}
\end{center}

The above table compares the task count between the unrestricted channels and the restricted channels.

\paragraph{Why the restriction matters.}
A small example shows how QMI violates DPI for general stochastic channels when \(n=3\). Let
\[
 U=\frac1{12}\begin{pmatrix}3&1\\0&4\\3&1\end{pmatrix},
 \qquad
 T=\begin{pmatrix}1&0&1\\0&1&0\\0&0&0\end{pmatrix}.
\]
Then \(T\) merges Alice's first and third reports and leaves the
second report unchanged. Direct substitution gives
\(\QMI(U)=1/12\) and \(\QMI(TU)=1/9\).
Thus this deviation raises expected payment by \(1/36\).
This \(T\) is excluded from \(\mathsf C_3\).

For \(2\times2\) channels, the doubly stochastic and rank-one classes
are the two diagonals in the parameter square
\(T=\bigl(\begin{smallmatrix}a&b\\1-a&1-b\end{smallmatrix}\bigr)\).
Their convex hull contains all stochastic channels. In fact, at \(n=m=2\),
\(\QMI(U)=4(\det U)^2\), which satisfies fixed-alphabet DPI for both sides.

\section{Discussion}\label{sec:conclusion}

A quantitative version of the zero-row
theorem could relate approximate DPI to boundary regularity, for
example through a modulus of continuity of the derivative. Such a
result would address how much information monotonicity can be retained
under finite-sample approximation. The present proofs are qualitative
and do not supply that tradeoff.

The determinantal-square theorem also leaves the remaining factors
unclassified. Even in the square case there is substantial freedom,
as the constant-compensation construction in
Appendix~\ref{app:constant-compensation-full-dpi} shows.
Appendix~\ref{app:hardness} further shows that recognizing DPI from
an arithmetic circuit is coNP-hard when the alphabet size is part of
the input. Describing
useful subclasses is a separate
problem from minimizing exact unbiased sample count.

The four-task construction in Section~\ref{sec:restricted} shows that
restricting channels permits positive rewards at every dependent
distribution even with unequal alphabets. Identifying larger channel
classes with useful incentive guarantees is a natural next question.

\section*{AI methodologies}
GPT-5.5 and GPT-5.6 assisted with parts of the research and preparation
of this paper. GPT-5.5 provided substantive assistance in developing the determinantal-rigidity proofs within a proof-strategy
framework supplied by the author. The zero-row argument was
developed by the author, with aspects of the approach inspired by
earlier interactions with GPT-5.5. GPT-5.6 subsequently assisted in
simplifying arguments and with the Lean~4 formalization. 

\section*{Formal verification and code availability}
The accompanying Lean~4 formalization uses Lean~4.32.0 and
mathlib~4.32.0. It covers the zero-row theorem, tall-case collapse, determinantal rigidity and its degree bounds. In particular, the Lean formalization of determinantal rigidity
includes a proof of the classical primaryness theorem for the
square of the generic maximal-minor ideal. This algebraic
development is independent of our application to dependence
measures and can be reused in other formalization projects. The source repository includes theorem mappings and
build instructions: \url{https://github.com/yuqkong-1844/finite-sample-information-monotones}.

\bibliographystyle{alphaurl}
\bibliography{ref}
\appendix
\section{An elementary determinant fact}
\label{app:determinant}

For completeness, we give a proof of the irreducibility argument used in the
square-case proof. 

\begin{lemma}[Irreducibility of the generic determinant]
\label{thm:generic-det-irreducible}
The determinant of the generic \(n\times n\) matrix over \(\bbR\) is
irreducible for every \(n\ge1\).
\end{lemma}

\begin{proof}
We argue by induction on \(n\). The case \(n=1\) is immediate.

Distinguish \(X:=U_{nn}\), and let \(A\) be the polynomial ring over
\(\bbR\) in all the remaining entries of the generic matrix. Then
\(D_n:=\det U\) may be viewed as a polynomial in \(X\) with
coefficients in \(A\):
\[
  D_n=aX+b,
\]
where \(a=D_{n-1}\) is the determinant of the upper-left generic
\((n-1)\times(n-1)\) block.

Recall that a polynomial ring over a field is a \emph{unique
factorization domain} (UFD), meaning that every nonzero polynomial
factors uniquely into irreducible factors, up to nonzero constants and
reordering. In a UFD, every irreducible element is prime: if it divides
a product, then it divides one of the factors. By induction,
\(a=D_{n-1}\) is irreducible in the polynomial ring generated by the
upper-left block. It remains prime after adjoining the other independent
variables appearing in \(A\), since the quotient by \((a)\) is then a
polynomial ring over an integral domain. Hence \(a\) is prime in \(A\).

We claim that \(a\) and \(b\) have no common nonconstant factor. Since
\(a\) is irreducible, it suffices to show \(a\nmid b\). Let
\(z:=U_{n,n-1}\). Viewed as a polynomial in \(z\), the coefficient of
\(z\) in \(b\), up to sign, is the determinant \(c\) of the submatrix
using rows \(1,\ldots,n-1\) and columns \(1,\ldots,n-2,n\). If
\(a\mid b\), then, since \(a\) does not involve \(z\), we would have
\(a\mid c\).

Now specialize the first \(n-2\) relevant rows and columns to an
identity block, set \(U_{n-1,n-1}=t\), set \(U_{n-1,n}=1\), and set
the remaining relevant entries to zero. Under this specialization,
\(a\mapsto t\) and \(c\mapsto1\), so \(a\mid c\) would imply
\(t\mid1\) in \(\bbR[t]\), a contradiction. Hence \(a\nmid b\), and
therefore \(a\) and \(b\) have no common nonconstant factor.

A polynomial whose coefficients have no common nonconstant factor is
called \emph{primitive}. Thus \(D_n=aX+b\) is primitive in \(A[X]\).

Let \(L:=\operatorname{Frac}(A)\) be the \emph{fraction field} of
\(A\): its elements are ratios \(p/q\) with \(p,q\in A\) and
\(q\ne0\), just as rational numbers are ratios of integers. In \(L\),
the nonzero coefficient \(a\) is invertible, so
\(D_n=a(X+b/a)\). Hence \(D_n\) is irreducible in \(L[X]\), since
\(X+b/a\) is a nonconstant linear polynomial.

Finally, Gauss's lemma in its standard UFD form
(see, e.g., Lang~\cite{lang02}) states that a primitive polynomial in
\(A[X]\) is irreducible over \(A\) if and only if it is irreducible
over the fraction field \(L\). Therefore \(D_n\) is irreducible in
\(A[X]\), and hence in the full polynomial ring of the generic
\(n\times n\) matrix.
\end{proof}

\section{A broad family of polynomial monotones}
\label{app:constant-compensation-full-dpi}

The divisibility theorem gives a necessary form, but does not classify
all admissible factors. The following construction shows how much
freedom remains. A polynomial \(P\in R_{n,n}\) is \emph{row-symmetric
on the simplex} if \(P(QU)=P(U)\) for every permutation matrix \(Q\)
and every \(U\in\Delta_{n,n}\).

\begin{proposition}[Constant compensation]
\label{prop:constant-compensation-full-dpi}
Let \(n\ge2\), and let \(P\in R_{n,n}\) be row-symmetric on the
simplex. There exists \(M\ge0\) such that
\[
 F_{P,M}(U)=(\det U)^2\bigl(P(U)+M\bigr)
\]
satisfies \eqref{eq:dpi} and \eqref{eq:iz} on \(\Delta_{n,n}\).
More explicitly, the quantity
\[
 M_*(P)=\sup_{\substack{T\text{ column-stochastic},\ (\det T)^2<1\\
                        U\in\Delta_{n,n}}}
 \frac{(\det T)^2P(TU)-P(U)}{1-(\det T)^2}
\]
is finite, and every \(M\ge\max\{0,M_*(P)\}\) suffices.
\end{proposition}

\begin{proof}
For every column-stochastic \(T\), Hadamard's inequality gives
\[
|\det T|
\le \prod_{j=1}^n\|T_{\bullet j}\|_2
\le 1,
\]
since each column is a probability vector. Equality holds exactly when
the columns are orthonormal probability vectors, hence exactly when
\(T\) is a permutation matrix. Thus \(0\le(\det T)^2\le1\), with
\((\det T)^2=1\) precisely for permutation matrices.

If \(T=Q\) is a permutation matrix, row-symmetry gives
\(P(QU)=P(U)\), and hence \(F_{P,M}(QU)=F_{P,M}(U)\). It therefore
suffices to consider \((\det T)^2<1\). Define
\[
R(T,U):=
\frac{(\det T)^2P(TU)-P(U)}
     {1-(\det T)^2}.
\]
We show that \(R\) is uniformly bounded above. Here
\(\|A\|_1:=\sum_{i,j}|A_{ij}|\) denotes the entrywise \(\ell^1\)-norm.

First consider \(T\) near a permutation matrix \(Q\). We claim that
there is a uniform constant \(C<\infty\) such that
\[
\|T-Q\|_1\le C\bigl(1-(\det T)^2\bigr)
\]
whenever \(T\) is column-stochastic and sufficiently close to \(Q\).
It suffices to prove this near \(I\): setting
\(\widetilde T=Q^{-1}T\) preserves column-stochasticity and gives
\[
 \|T-Q\|_1=\|\widetilde T-I\|_1,
 \qquad
 (\det T)^2=(\det\widetilde T)^2.
\]
Thus the reduction also applies to odd permutations, for which
\(\det T\) is near \(-1\).

For \(T\) close to \(I\), put
\(\varepsilon:=\sum_j\sum_{i\ne j}T_{ij}\). Column-stochasticity gives
\(T_{jj}-1=-\sum_{i\ne j}T_{ij}\), and therefore
\(\|T-I\|_1=2\varepsilon\). Writing \(T=I+A\), the Taylor expansion
\(\det(I+A)=1+\operatorname{tr}(A)+O(\|A\|_1^2)\) and
\(\operatorname{tr}(A)=-\varepsilon\) yield
\(\det T=1-\varepsilon+O(\varepsilon^2)\). After shrinking the
neighborhood, \(1-\det T\ge\varepsilon/2\) and \(\det T\ge0\).
Consequently,
\[
1-(\det T)^2
=(1-\det T)(1+\det T)
\ge \varepsilon/2,
\]
so \(\|T-I\|_1=2\varepsilon\le4(1-(\det T)^2)\). Since there are only
finitely many permutation matrices, the same estimate holds near all
of them with one constant \(C\).

Because \(P\) is polynomial and \(\Delta_{n,n}\) is compact, there
are constants \(L,B<\infty\) such that
\[
|P(V)-P(W)|\le L\|V-W\|_1,
\qquad
|P(V)|\le B
\]
for all \(V,W\in\Delta_{n,n}\). If \(T\) lies near a permutation
matrix \(Q\), then \(P(QU)=P(U)\), so
\[
\begin{aligned}
(\det T)^2P(TU)-P(U)
&=(\det T)^2\bigl(P(TU)-P(QU)\bigr)
  +\bigl((\det T)^2-1\bigr)P(U)\\
&\le L\|(T-Q)U\|_1+B\bigl(1-(\det T)^2\bigr)\\
&\le (LC+B)\bigl(1-(\det T)^2\bigr).
\end{aligned}
\]
Here \(\|(T-Q)U\|_1\le\|T-Q\|_1\), since \(U\) is nonnegative with
total mass one. Hence \(R(T,U)\le LC+B\) uniformly near the
permutation matrices.

It remains to consider \(T\) away from them. Let \(\mathcal S\) be
the compact set of column-stochastic \(n\times n\) matrices, and let
\(\mathcal S_0\) be the complement in \(\mathcal S\) of the
neighborhoods just chosen. Since \((\det T)^2=1\) only at permutation
matrices, compactness gives \(\eta>0\) such that
\(1-(\det T)^2\ge\eta\) on \(\mathcal S_0\). Therefore
\[
R(T,U)
\le
\frac{(\det T)^2|P(TU)|+|P(U)|}
     {1-(\det T)^2}
\le \frac{2B}{\eta}.
\]
Thus
\[
M_*(P):=
\sup_{\substack{T\in\mathcal S,\ U\in\Delta_{n,n}\\(\det T)^2<1}}
R(T,U)
<\infty.
\]
Choose \(M\ge\max\{0,M_*(P)\}\). Then, whenever
\((\det T)^2<1\),
\[
(\det T)^2P(TU)-P(U)
\le M\bigl(1-(\det T)^2\bigr),
\]
or equivalently
\[
(\det T)^2\bigl(P(TU)+M\bigr)\le P(U)+M.
\]
For permutation matrices the same relation holds with equality.
Multiplying by \((\det U)^2\ge0\) and using
\(\det(TU)=\det T\,\det U\) gives
\[
F_{P,M}(TU)\le F_{P,M}(U),
\]
as required. Independence-zero follows from the determinant-square
factor, since \(n\ge2\).
\end{proof}

\section{Normalization under fixed-alphabet DPI and independence tests}
\label{sec:normalization}

We state the one-sided and two-sided normalization results
separately. Throughout, \(n,m\ge2\), and
\(\Delta_m:=\{b\in\mathbb R_{\ge0}^m:\ones_m^\top b=1\}\).

\begin{proposition}[One-sided normalization and structure]
\label{prop:one-sided-normalization}
Let \(F:\Delta_{n,m}\to\mathbb R\) satisfy \eqref{eq:dpi}.
Define
\[
 \varphi(b):=F(e_1b^\top),\qquad
 F_0(U):=F(U)-\varphi(c(U)),
\]
where \(b\in\Delta_m\) and \(e_1\in\mathbb R^n\) is the
first standard basis vector. Then:
\begin{enumerate}
\item For every pair of probability vectors \(a,b\),
\(F(ab^\top)=\varphi(b)\).
The functional \(F_0\) is nonnegative and satisfies
\eqref{eq:dpi} and \eqref{eq:iz}.
The normalization preserves \(C^1\)-extendability
and polynomiality.

\item If \(F\) is \(C^1\)-extendable, then
\(F(U)=\varphi(c(U))\) whenever one row of \(U\) is a
nonnegative linear combination of the others.
In particular, this holds on every zero-row face.
If \(n>m\), it holds throughout \(\Delta_{n,m}\).

\item If \(F\) is polynomial and \(n\le m\), then
\[
 F_0\in I_n^2+(\tau-1).
\]
For \(n=m\), this gives
\[
 F(U)=\varphi(c(U))+(\det U)^2G(U)
\]
on the simplex, for some polynomial \(G\).
If \(F_0\not\equiv0\), every polynomial representing \(F\)
on the simplex has degree at least \(2n\), and every
fixed-sample unbiased estimator of \(F\) uses at least
\(2n\) samples. These bounds are sharp.
\end{enumerate}
\end{proposition}

\begin{proof}
The column-stochastic channels \(e_1\ones_n^\top\) and
\(a\ones_n^\top\) map \(ab^\top\) to \(e_1b^\top\)
and back. Applying \eqref{eq:dpi} in both directions gives
\(F(ab^\top)=\varphi(b)\), hence independence-zero for
\(F_0\). Since \(c(TU)=c(U)\), subtracting
\(\varphi(c(U))\) preserves \eqref{eq:dpi}.
Nonnegativity follows from Lemma~\ref{lem:refPone-derived}.
Moreover,
\[
 \varphi(c(U))=F(e_1\ones_n^\top U).
\]
Composition with this linear map preserves polynomiality.
It also preserves \(C^1\)-extendability, since the map sends
the simplex into itself.

Apply Lemma~\ref{lem:row-cone-vanishing} and
Theorem~\ref{thm:normalized} to \(F_0\) to obtain the second
claim. Theorem~\ref{thm:non-tall} gives the ideal membership
and degree bound in the third claim, and
Corollary~\ref{thm:nxn} gives the square specialization.
For any polynomial representative \(P\) of \(F\),
\[
 P(U)-P(e_1\ones_n^\top U)
\]
represents \(F_0\) and has degree at most \(\deg P\).
Thus \(F_0\not\equiv0\) implies \(\deg P\ge2n\).
The sample bound follows from the polynomial
characterization of fixed-sample unbiased estimation.
Sharpness is witnessed by \(\det(UU^\top)\).
\end{proof}

\begin{proposition}[Two-sided normalization and structure]
\label{prop:normalization}
Let \(F:\Delta_{n,m}\to\mathbb R\) satisfy
\eqref{eq:two-sided-fixed-alphabet-dpi}. Then:
\begin{enumerate}
\item All independent distributions have a common value
\(\gamma\), and \(F-\gamma\) is nonnegative and satisfies
\eqref{eq:iz} and
\eqref{eq:two-sided-fixed-alphabet-dpi}.
Subtracting \(\gamma\) preserves \(C^1\)-extendability
and polynomiality.

\item If \(F\) is \(C^1\)-extendable and \(n\ne m\),
then \(F\equiv\gamma\).

\item If \(F\) is polynomial and \(n=m\), then
\[
 F(U)=\gamma+(\det U)^2G(U)
\]
on the simplex, for some polynomial \(G\).
If \(F\) is nonconstant, every polynomial representing it
has degree at least \(2n\), and every fixed-sample unbiased
estimator uses at least \(2n\) samples.
These bounds are sharp.
\end{enumerate}
\end{proposition}

\begin{proof}
Any independent distribution \(ab^\top\) can be sent to
any other \(a'(b')^\top\) using the constant-output channels
\(a'\ones_n^\top\) and \(b'\ones_m^\top\).
The reverse transformation is also possible, so two-sided
DPI gives a common value \(\gamma\).
The normalization claims follow immediately.

If \(n>m\), apply Theorem~\ref{thm:normalized} to
\(F-\gamma\). If \(m>n\), apply it to the transposed
functional. In either case, \(F\equiv\gamma\).
For \(n=m\), Corollary~\ref{thm:nxn} gives the
determinant-square factorization. The degree and sample
bounds follow from Theorem~\ref{thm:non-tall};
sharpness is witnessed by \((\det U)^2\).
\end{proof}

\paragraph{An implication for independence tests.}
Let \(\phi:([n]\times[m])^N\to[0,1]\) be a randomized test, and
let \(R_\phi(U)=\mathbb E_{U^{\otimes N}}[\phi]\) be its rejection
probability. This is a polynomial in \(U\). If its entire power
function is required to satisfy \eqref{eq:two-sided-fixed-alphabet-dpi}, then for \(n\ne m\)
Proposition~\ref{prop:normalization} forces \(R_\phi\) to be constant.
Such a test cannot distinguish independence from dependence through
its rejection probability. This conclusion concerns an additional
monotonicity requirement on the power function; ordinary independence
testing does not impose that requirement.

\section{The complexity of recognizing DPI}
\label{app:hardness}

The determinant-square factor is necessary for polynomial DPI
monotones, but it does not characterize the admissible remaining
factors. Appendix~\ref{app:constant-compensation-full-dpi}
exhibits a broad sufficient family. Here we show, by a reduction
from copositivity testing, that recognizing DPI for
circuit-represented polynomials is coNP-hard when the alphabet
size is part of the input. This establishes a computational
obstacle to efficiently recognizing the full class.

For \(U\in\Delta_{n,n}\), let \(c(U)=U^\top\ones_n\) be the
column marginal, and write
\(\Delta_n=\{c\in\mathbb R_{\ge0}^n:\ones_n^\top c=1\}\).

\begin{theorem}[DPI recognition is coNP-hard]
\label{thm:dpi-recognition-hard}
When \(n\ge2\) is part of the input, deciding whether a rational polynomial
\(F\), represented by an arithmetic circuit, satisfies \eqref{eq:dpi}, i.e., 
\[
 F(TU)\le F(U)
 \quad\text{for every }U\in\Delta_{n,n}
 \text{ and every column-stochastic }T\in\mathbb R^{n\times n}
\]
is coNP-hard. The hardness holds for the family
\[
 F_A(U)=(\det U)^2c(U)^\top A c(U),
 \qquad A\in\mathbb Q^{n\times n},\quad A=A^\top.
\]
Every polynomial in this family is row-permutation invariant and
independence-zero.
\end{theorem}
\begin{proof}
A symmetric matrix \(A\) is \emph{copositive} if
\(x^\top Ax\ge0\) for all \(x\ge0\). Homogeneity makes this equivalent
to \(c^\top Ac\ge0\) for all \(c\in\Delta_n\). Testing copositivity
is coNP-complete~\cite{MurtyKabadi1987}.

Given \(A\), construct \(F_A\) as in the statement. This takes
polynomial time in the binary encoding of \(A\): column sums and a
quadratic form have polynomial-size arithmetic circuits, as does the
determinant over the rationals. For example, Newton identities applied
to the traces of \(U,U^2,\ldots,U^n\) give such a circuit using only
division by fixed nonzero integers, which can be encoded as rational
constants.

Column stochasticity implies
\[
 c(TU)=U^\top T^\top\ones_n=U^\top\ones_n=c(U).
\]
Therefore
\begin{equation}\label{eq:hardness-gap}
 F_A(TU)-F_A(U)
 =\bigl((\det T)^2-1\bigr)(\det U)^2c(U)^\top Ac(U).
\end{equation}
If \(A\) is copositive, the last quadratic form is nonnegative and
Hadamard's inequality gives \(|\det T|\le1\); hence \(F_A\) satisfies
\eqref{eq:dpi}.

Conversely, if \(A\) is not copositive, choose
\(c\in\Delta_n\) with \(c^\top Ac<0\). By continuity, mixing \(c\)
with a sufficiently small positive multiple of the uniform
distribution preserves strict negativity, so we may assume \(c>0\).
For \(U=\operatorname{diag}(c_1,\ldots,c_n)\), we have
\(\det U>0\), \(c(U)=c\), and \(F_A(U)<0\).
Any rank-one stochastic channel \(T_0\) has determinant zero and
therefore
\[
 F_A(T_0U)=0>F_A(U).
\]
Thus \(F_A\) satisfies \eqref{eq:dpi} exactly when \(A\) is copositive, proving
the reduction. Row-permutation invariance follows from invariance of
\(c(U)\) and of \((\det U)^2\), while rank one implies determinant
zero because \(n\ge2\).
\end{proof}

\end{document}